\documentclass[11pt]{article}

\usepackage{hyperref}
\usepackage{graphicx}
\usepackage{amsmath}
\usepackage{amsthm}
\usepackage[overload]{empheq}
\usepackage{amssymb}
\usepackage{amsfonts}
\usepackage{epsfig}
\usepackage{verbatim}
\usepackage{bm}
\usepackage{newtype}
\usepackage{pst-all}

\usepackage{lineno}
\newtheorem{theorem}{Theorem}[section]
\newtheorem{lemma}[theorem]{Lemma}
\newtheorem{corollary}[theorem]{Corollary}

\newtheorem{remark}[theorem]{Remark}

\theoremstyle{definition}

\makeatletter
\@addtoreset{claim}{theorem}
\makeatother

\newcounter{problemctr}[subsection]

\newcommand{\NN}{\mathbb{N}}
\newcommand{\ZN}{\mathbb{Z}}
\newcommand{\QN}{\mathbb{Q}}
\newcommand{\RN}{\mathbb{R}}

\newcommand{\norm}[1]{\lVert #1 \rVert}

\newtype{\class}{\mathbf}
\class{P}
\class{NP}
\class{AM}
\class{LOGNP}
\class{coNP}
\class{NTIME}
\class{PSPACE}
\class{EXP}
\class{EXPSPACE}
\class{NEXP}
\class{RP}
\class{coRP}
\class{TA}
\class{RA}
\class{PP}
\class{PH}
\class{CH}
\class{FIXP}
\class[FIXPd]{FIXP_d}
\class[sharpP]{\#P}

\newcommand{\ER}{\class<\ensuremath{\bm{\exists \RN}}>}
\newcommand{\VR}{\class<\ensuremath{\bm{\forall \RN}}>}
\newcommand{\VER}{\class<\ensuremath{\bm{\forall\exists \RN}}>}
\newcommand{\EVR}{\class<\ensuremath{\bm{\exists\forall \RN}}>}
\newcommand{\VEVR}{\class<\ensuremath{\bm{\forall\exists\forall \RN}}>}

\DeclareMathOperator{\conv}{conv}

\def\RCC8{\textup{RCC8}}

\newcommand{\Krass}{{Krasnosel\textquotesingle ski\u{\i}}}

\title{How Hard is it to be a Star?\\Convex Geometry and the Real Hierarchy}

\author{
{Marcus Schaefer
} \\
{\small School of Computing} \\[-0.13cm]
{\small DePaul University} \\[-0.13cm]
{\small Chicago, Illinois 60604, USA} \\[-0.13cm]
{\small \tt mschaefer@cdm.depaul.edu}\\[-0.13cm]
\and
{Daniel \v{S}tefankovi\v{c}
} \\
{\small Computer Science Department} \\[-0.13cm]
{\small University of Rochester} \\[-0.13cm]
{\small Rochester, NY 14627-0226} \\[-0.13cm]
{\small \tt stefanko@cs.rochester.edu}\\[-0.13cm]
}

\begin{document}

\maketitle

\begin{abstract}
A set in $\RN^d$ is {\em star-shaped} if there is a point in the set that can {\em see} every other point in the set in the sense that the line-segment connecting the points lies within the set. We show that testing whether a non-empty compact smooth region is star-shaped is \VR-complete. Since the obvious definition of star-shapedness has logical form $\exists\forall$, this is a somewhat surprising result, based on \Krass's theorem from convex geometry; we study several related complexity classifications in the real hierarchy based on other results from convex geometry.
\end{abstract}



\section{Introduction}

In an earlier paper~\cite{SS23,SS25} we studied the computational complexity of properties of semialgebraic sets that can be classified at various levels of the real polynomial hierarchy, such as \ER, \VR, \EVR, and \VER\ (for background on the existential theory of the reals and the real polynomial hierarchy, see~\cite{M14,SCM24,SS25,HMMP25}). One of the examples we encountered behaved unexpectedly: the radius problem.

A semialgebraic set $S$ has {\em radius at most $r$} if there is a point $c$, the {\em center}, such that all points in $S$ have distance at most $r$ from $c$. More formally,
\[(\exists c \in \RN^d)(\forall x \in \RN^d)\ x \in S \rightarrow \sum_{i=1}^d (x_i-c_i)^2 \leq r^2.\]
So the straightforward definition of radius gives us an \EVR-characterization of deciding whether $S$ has radius at most $r$. Using Helly's theorem, we were able to obtain a \VER-characterization saying that any $d+1$ balls of radius $r$ with centers in $S$ have a common point, placing the radius problem in $\EVR \cap \VER$. Using linear programming duality, the non-emptiness of the intersection of the $d+1$ balls can be expressed in \VR, leading to a \VR-characterization of the radius problem (for details, see~\cite[Theorem 3.3]{SS23}). Since the problem is easily seen to be \VR-hard, this settles the complexity of the radius problem for semialgebraic sets.

The radius example inspired us to look for further instances of this effect, in which the combination of a result from convex geometry with duality leads to a surprisingly lower complexity classification than what one would have expected based on the definition of the problem.

Our main result in the current paper is on star-shapedness. A set is {\em star-shaped} if it contains a point $c$ that {\em sees} all points in the set in the sense that the line-segments connecting $c$ to any point in the set must belong to the set. We determine the complexity of testing star-shapedness for semialgebraic sets whose boundary is sufficiently smooth; we call these sets smooth regions (for the formal definition see Section~\ref{sec:starER}).

\begin{theorem}\label{thm:starVR}
  Deciding whether a non-empty compact smooth region 
  is star-shaped is \VR-complete.
\end{theorem}

The result from convex geometry behind the \VR-membership in Theorem~\ref{thm:starVR} is \Krass's characterization of star-shapedness. As part of the \VR-hardness argument we establish a new \ER-hardness result which may be of independent interest: testing whether a polynomial contains a zero in the unit ball is \ER-complete, even if the polynomial is strictly positive outside the unit ball, see Theorem~\ref{thm:bdfeasstrong} for details.

\begin{remark}[Assumptions]
\Krass's characterization needs compactness, which is why we require compactness in Theorem~\ref{thm:starVR}. We also require the set to be non-empty, since testing non-emptiness is \ER-complete, and empty sets are not star-shaped by definition; so Theorem~\ref{thm:starVR} would not hold without the non-emptiness assumption (unless $\ER = \VR$), or we redefine empty sets to be star-shaped (which is not done traditionally). Therefore, in the language of complexity theory, Theorem~\ref{thm:starVR} is about star-shapedness as a promise problem, rather than as a decision problem: we are making assumptions about the input we cannot verify within the complexity class within which we classify the problem~\cite{G06}.

On the other hand, being bounded is \VR-complete~\cite[Proposition 6.4,Corollary 9.4]{BC09}. The complexity of being closed for a semialgebraic set is unknown; for a basic semialgebraic set, it is known to be \VR-complete~\cite[Theorem 6.15]{BC09}; since our smooth regions are basic semialgebraic sets it follows that testing compactness for smooth regions lies in \VR. So we could modify Theorem~\ref{thm:starVR} to state that testing whether a non-empty smooth region is star-shaped and compact is \VR-complete.  
\end{remark}

\subsection{Context and Overview}

There are quite a number of results already for which the usual definition of a problem does not agree with its eventual complexity classification. Many of these examples involve properties of semialgebraic sets, and go back to papers by Cucker, Rossell\'{o}~\cite{CR92}, Koiran~\cite{K99,K00}, and B\"{u}rgisser, Cucker~\cite{BC09}.

\begin{description}
  \item[Unboundedness.] A semialgebraic set $S$ is {\em unbounded} if for every $r > 0$ there is an $x \in S$ such that $\norm{x} > r$, so the obvious complexity upper bound on unboundedness is \VER, but the universal quantifier can be replaced with
  the exotic quantifier $H$ introduced by Koiran~\cite{K99}; $H$ quantifies over all sufficiently small values of a variable. The quantifier $H$ can be eliminated in this case, leading to an upper bound of \ER, shown in ~\cite[Corollary 9.4]{BC09} extending~\cite{CR92}. Other examples that depend on eliminating an $H$-quantifier include testing whether a polynomial or a semialgebraic set has an isolated zero~\cite{BC09,SS17}, testing whether a basic semialgebraic set is closed or compact~\cite{BC09}, the distance between two semialgebraic sets~\cite{SS17}, and the angular resolution of a graph~\cite{S23b}.
  \item[Dimension.] A semialgebraic set $S \subseteq \RN^d$ has dimension $d$ if there is a point $x \in \RN^d$ and an $\varepsilon > 0$
  so that all $y$ of distance at most $\varepsilon$ from $x$ lie in $S$. By this definition the dimension problem lies in \EVR, but Koiran~\cite{K99} saw that the problem can be expressed using the exotic quantifier $\exists^*$, which expresses that a property holds for an open set; moreover, he showed that the exotic quantifier can be replaced with standard existential quantifiers, eventually placing the dimension problem into \ER, making it \ER-complete. Other examples that can be captured using $\exists^*$ or $\forall^*$ (expressing ``for a dense set''), include density properties of semialgebraic sets~\cite{BC09}; a hyperplane locally supporting a semialgebraic set requires elimination of both $H$ and $\exists^*$ quantifiers~\cite{BC09}.
\end{description}

There are also examples not based on exotic quantifiers:

\begin{description}
  \item[Art Gallery Problem.] An art gallery (a polygon) can be guarded by $k$ guards if there is a placement of the $k$ guards in the art gallery so that every point in the art gallery can be seen by at least one of the guards. This definition gives an \EVR-characterization of the art gallery problem, but the inner two universal quantifiers can be removed by using the linear structure of the polygon~\cite{AAM22}. This implies that the art gallery problem lies in \ER, and it is actually \ER-complete, as shown by Abrahamsen, Adamaszek, Miltzow~\cite{AAM22}, and Stade~\cite{S23c}. 
\end{description}

The examples in this paper are different in that they rely on duality and characterizations from convex geometry by Carath\'{e}odory, Steinitz, Kirchberger, and \Krass\ to flip quantifiers. The only previous result of this nature we are aware of is the radius problem we discussed above.

Section~\ref{sec:CS} illustrates our approach with some easy examples: we first show that convex hull membership lies in \ER; this example only requires Ca\-ra\-th\'{e}o\-dory's theorem, and no duality. Pushing slightly farther, membership in the interior of the convex hull also lies in \ER, and for that result we combine duality with Steinitz's theorem, which extends Carath\'{e}odory.

In Section~\ref{sec:Kirch} we turn to Kirchberger's theorem and separability by hyperplanes, which turns out to be \VR-complete, rather than the obvious \EVR. Finally, Section~\ref{sec:Krass} applies \Krass-theorem to shar-shapedness. We conclude the paper with a section on open problems.

Table~\ref{tbl:prop} summarizes complexity results in the real hierarchy depending on results from convex geometry.

\begin{table}[!htb]
\centering
\begin{tabular}{ |l|l|l|l| } 
\hline
    problem & complexity & reference & based on \\ \hline\hline
    convex hull & \ER-complete & Theorem~\ref{thm:inconvER} & Carath\'{e}odory \\ \hline
    interior convex hull & \ER-complete & Corollary~\ref{cor:inintconvER} & Steinitz \\ \hline
    (interior of) positive hull & \ER-complete & Corollary~\ref{cor:inposER} & Carath\'{e}odory \\ \hline
    \parbox{1.4in}{(strict) hyperplane \\separability} & \VR-complete & \parbox{1in}{Theorem~\ref{thm:strictsepVR},\\ Corollary~\ref{cor:sepVR}} & Kirchberger \\ \hline
    star-shapedness$^*$ & \VR-complete & Theorem~\ref{thm:starVR} & \Krass \\ \hline
    radius & \VR-complete & \cite[Theorem 3.3]{SS23} & Helly \\ \hline
\end{tabular}
\caption{Properties of semialgebraic sets and their complexity based on results from convex geometry. The starred result for star-shaped sets is  valid for non-empty, compact smooth regions only.}\label{tbl:prop}
\end{table}

\subsection{Duality}

Several of our proofs will make use of linear programming duality. We will use this type of duality in two forms known as Fredholm's alternative~\cite[Corollary 3.1b]{S86} and Farkas' lemma~\cite{F02}, also~\cite[Section 7.1]{S86}. 

\begin{lemma}[Fredholm]\label{lem:fredholm}
 Let $A \in \RN^{m\times n}$ and $b \in \RN^m$. Then the system $Ax = b$ is solvable if and only if there is no $v\in \RN^n$ with $A^Tv = 0$ and $b^Tv \neq 0$
 (equivalently, for all $v\in \RN^n$ either $A^T v\neq 0$ or $b^T v =0$). 
\end{lemma}

The lemma is known as Fredholm's alternative since it implies that either $Ax = b$ or $A^Tv = 0,
\ b^Tv = 1$ is solvable. Fredholm's lemma is generalized by Farkas lemma which captures inequalities.

%A more general version of duality deals with inequalities.

\begin{lemma}[Farkas] \label{lem:farkas}
Let $A\in{\mathbb{R}}^{m\times n}$ and $b\in{\mathbb{R}}^m$. Then the system $Ax\geq b$ is solvable if and only if there is no $v\in\RN^n\geq 0$
with $A^T v = 0$ and $b^T v > 0 $  (equivalently, for all $v\in \RN^n_{\geq 0}$ either $A^T v\neq 0$ or $b^T v \leq 0$). 
\end{lemma}

\section{Carath\'{e}odory, Steinitz and the Convex Hull}\label{sec:CS}

We say a point $p$ is the {\em convex combination} of a set of points $p_1, \ldots, p_n$ if $p = \sum_{i\in [n]} \lambda_i p_i$
with $\sum_{i \in [n]} \lambda_i = 1$ and $\lambda_i \geq 0$ for all $i \in [n]$. We also say that $p$ lies in the {\em convex hull} $\conv(p_1, \ldots, p_n)$ of $p_1, \ldots, p_n$.
Carath\'{e}odory proved the following theorem~\cite[Theorem 1.21]{V64}.

\begin{theorem}[Carath\'{e}odory]\label{thm:Cthm}
 A point belongs to the convex hull of a set in $\RN^d$ if and only if it is the convex combination of at most $d+1$
 points of the set.
\end{theorem}

With this we can settle the complexity of membership in the convex hull of a semialgebraic set.

\begin{theorem}\label{thm:inconvER}
 Deciding whether a point $q \in \QN^d$ lies in the convex hull of a semialgebraic set is \ER-complete.
\end{theorem}

For \ER-hardness we reduce from a problem on {\em integer polynomials}, that is, polynomials with integer coefficients. The {\em bounded polynomial feasibility} problem asks whether for a given integer polynomial $f: \RN^d \rightarrow \RN$, there is an $x \in B^d(0,1) := \{x \in \RN^d:\ \norm{x} \leq 1\}$ such that $f(x) = 0$? The bounded feasibility problem is \ER-complete~\cite[Lemma 3.9 with $f = \sum f^2_i$]{S13}; the proof shows that the problem remains \ER-complete if $f$ is non-negative, has total degree at most $4$ and either $f$ is positive if it does not have a zero in $B(0,1)$. 

\begin{proof}
To see \ER-hardness let $f: \RN^d \rightarrow \RN$ be an instance of the bounded polynomial feasibility problem.
Define $S = \{-x:\ f(x) = 0, x \in B^d(0,1)\} \cup \{x:\ f(x) = 0, x \in B^d(0,1)\}$. If $f$ has a zero in $B^d(0,1)$, then $0$ belongs to the convex hull of $S$, otherwise $S$ is empty.

To argue membership, let $S$ be the semialgebraic set defined by formula $\varphi$, so $S = \{x \in \RN^d:\ \varphi(x)\}$. Using Theorem~\ref{thm:Cthm} we can express that $q$ lies in the convex hull of $S$ as
\begin{align*}
(\exists a_1, &\ldots, a_{d+1} \in \RN^d)(\exists \lambda_1, \ldots, \lambda_{d+1} \in \RN_{\geq 0}) \\
        & \bigwedge_{i \in [d+1]} \varphi(a_i) \wedge \sum_{i \in [d+1]} \lambda_i a_i = q \wedge \sum_{i \in [d+1]} \lambda_i = 1 .
\end{align*}
This establishes membership in \ER.
\end{proof}

Steinitz proved an extension of Carath\'{e}odory's result that captures the interior of the convex hull~\cite[Theorem 10.3]{E93}.

\begin{theorem}[Steinitz]\label{thm:Sthm}
 A point lies in the interior of the convex hull of a set in $\RN^d$ if and only if it lies in the interior of the convex hull of at most $2d$ points in the set.
\end{theorem}

The bound of $2d$ is known to be optimal: the $2d$ points $\{\pm e_1, \ldots, \pm e_d\}$, where $e_1, \ldots, e_d$ is the standard basis of $\RN^d$, contain $0$ in their convex hull, but this is not true for any proper subset of these points. With Steinitz's version of Carath\'{e}odory's theorem, we can show how to recognize the interior of a convex hull in \ER.

\begin{corollary}\label{cor:inintconvER}
 Deciding whether a point $q \in \QN^d$ lies in the interior of the convex hull of a semialgebraic set is \ER-complete.
\end{corollary}

Call a convex combination {\em strict} if all the coefficients $\lambda_i$ are positive. Then a point $q$ lies in the interior of the convex hull of $2d$ points $a_1, \ldots a_{2d}$ if it is a strict convex combination of these points, and the affine space spanned by the points has full dimension, that is, they do not lie in a common hyperplane. Formally, we can write this as
\begin{equation}\label{eq:inintconv}
\begin{aligned}
(\exists \lambda_1, &\ldots, \lambda_{2d} \in \RN_{> 0}) \\
        & \sum_{i \in [2d]} \lambda_i a_i = q \wedge \sum_{i \in [2d]} \lambda_i = 1\ \wedge \\
        & \neg (\exists u \in \RN^d)\ u \neq 0 \wedge \bigwedge_{i \in [2d-1]} (a_i-a_{2d}) \cdot u = 0.
\end{aligned}
\end{equation}

The subformula starting with the negated existential quantifier in the last line verifies that the $a_i$ do not lie on a common hyperplane by expressing that there is no non-trivial normal vector $u$. Since $\neg \exists$ is equivalent to $\forall \neg$, we obtain a universal characterization of the hyperplane condition. We can then use linear programming duality in the form of Fredholm's alternative to exchange universal with existential quantification. Let us work out the details. Consider the existential condition we are using for lying on a common hyperplane:
\begin{equation}\label{eq:inhyp}
(\exists u \in \RN^d)\ u \neq 0 \wedge \bigwedge_{i \in [2d-1]} (a_i-a_{2d}) \cdot u = 0.
\end{equation}
We can rewrite $u \neq 0$ as $\bigvee_{j \in [d]} u_j \neq 0$, and, since we can scale $u$, as $\bigvee_{j \in [d]} u_j = 1$. Exchanging the order of the quantifier and the disjunction we get that~\eqref{eq:inhyp} is equivalent to
\begin{equation}\label{eq:inhyp2}
 \bigvee_{j \in [d]} (\exists u \in \RN^d)\  u_j = 1 \wedge \bigwedge_{i \in [2d-1]} (a_i-a_{2d}) \cdot u = 0.
\end{equation}
Define $d$ matrices $A_j = ((a_1-a_{2d})^T \vert \cdots \vert (a_{2d-1}-a_{2d})^T \vert e^T_j)$, for $j \in [d]$, where $e_j$ is the $j$-th vector in the standard basis of $\RN^d$ and $|$ denotes horizontal concatenation (of column vectors). Then~\eqref{eq:inhyp2} can be written as 
\begin{equation}\label{eq:inhyp2.5}
\bigvee_{j \in [d]} (\exists u \in \RN^d)\  A_ju = (0^d \vert 1)^T.
\end{equation}
 Fredholm's alternative, Lemma~\ref{lem:fredholm}, then allows us to conclude that $(\exists u \in \RN^d)\  A_ju = (0^d \vert 1)^T$ fails if and only if
$(\exists v \in \RN^{d+1})\ A_j^Tv = 0 \wedge (0^d \vert 1) \cdot v \neq 0$ succeeds; that last formula simplifies to $(\exists v \in \RN^{d+1})\ A_j^Tv = 0 \wedge v_{d+1} \neq 0$.
This allows us to rewrite~\eqref{eq:inintconv} as
\begin{equation}\label{eq:inintconv2}
\begin{aligned}
(\exists \lambda_1, &\ldots, \lambda_{2d} \in \RN_{> 0}) \\
        & \sum_{i \in [2d]} \lambda_i a_i = q \wedge \sum_{i \in [2d]} \lambda_i = 1\ \wedge \\
        & \bigwedge_{j \in [d]} (\exists v \in \RN^{d+1})\ A_j^Tv = 0 \wedge v_{d+1} \neq 0,
\end{aligned}
\end{equation}
which is purely existential.

\begin{proof}[Proof of Corollary~\ref{cor:inintconvER}]
We follow the proof of Theorem~\ref{thm:inconvER}. To see \ER-hardness we again reduce from the bounded polynomial feasibility problem. Given polynomial $f: \RN^d \rightarrow R$, let
\begin{equation*}
\begin{aligned}
S = & \bigcup_{i \in [d]} \{x+2e_i:\ f(x) = 0, x \in B^d(0,1)\}\ \cup \\
    & \bigcup_{i \in [d]} \{x-2e_i:\ f(x) = 0, x \in B^d(0,1)\}.
\end{aligned}
\end{equation*}

If $f$ has a zero in $B^d(0,1)$, then $S$ contains $0$ in its interior, otherwise $S$ is empty.

For \ER-membership we adapt the proof in 
Theorem~\ref{thm:inconvER}, 
replacing Ca\-ra\-th\'{e}\-o\-do\-ry's Theorem~\ref{thm:Cthm} with Steinitz's Theorem~\ref{thm:Sthm}. We test whether a point lies in the interior of a set of $2d$ points in $\RN^d$ using \eqref{eq:inintconv2}.
\end{proof}

Steinitz's theorem extends to relative interiors, see~\cite[Theorem 10.4]{E93}, so testing membership in relative interiors is likely \ER-complete as well.

\medskip

We can also extend the argument in Theorem~\ref{thm:inconvER} to apply to the positive (or conic) hull of a set. The {\em positive hull} of a set $S \subseteq \RN^d$ is defined as $\operatorname{pos}(S) = \{\sum_{i \in [n]} \lambda_i a_i:\ a_i \in S, \lambda_i \in \RN_{\geq 0}, n \in \NN\}$, the set of all non-negative linear combinations of elements in $S$. Caratheodory's theorem for cones~\cite[Theorem 7.1i]{S86} implies that any element of $\operatorname{pos}(S)$ is a non-negative linear combination of at most $d$ elements of $S$, using $S \subseteq \RN^d$. Consequently, any element of the interior of $\operatorname{pos}(S)$ is a strict combination of at most $d$ elements of $S$.

\begin{corollary}\label{cor:inposER}
 Deciding whether a point $q \in \QN^d$ lies in the (interior of the) positive hull of a semialgebraic set is \ER-complete.
\end{corollary}

\section{Kirchberger and Hyperplane Separations}\label{sec:Kirch}

We say two sets $A, B \in \RN^d$ are {\em (strictly) separated by a hyperplane} $H$ if there is a $v \in \RN^d\setminus \{0\}$ and $c \in \RN$ such that
$v \cdot a \leq c$ and $v \cdot b \geq c$ ($v \cdot a < c$ and $v \cdot b > c$) for all $a \in A, b \in B$, where $\cdot$ is the dot-product of two vectors. We will treat strict separability first and then point out the necessary modifications for separability.

For semialgebraic sets $A = \{x \in \RN^d:\ \alpha(x)\}$ and $B = \{x \in \RN^d:\ \beta(x)\}$ we can express being strictly separable by a hyperplane as
\[(\exists v \in \RN^d)(\exists c \in \RN)(\forall a \in A)(\forall b \in B)\ v \neq 0 \wedge v\cdot a < c \wedge v \cdot b > c,\]
in other words, the problem lies in \EVR. Enter a theorem proved by Paul Kirchberger, a student of David Hilbert.\footnote{The little we know about Kirchberger can be found in~\cite[Chapter 4]{S06}.}

\begin{theorem}[Kirchberger~\cite{K03}]\label{thm:KBthm}
 If $A$ and $B$ are compact sets in $\RN^d$ then $A$ and $B$ can be strictly separated by a hyperplane if and only if $A \cap P$ and $B \cap P$ can be strictly separated by a hyperplane for every set $P$ of $d+2$ points in $A \cup B$.
\end{theorem}

A proof of Kirchberger's theorem can be found in Valentine's book~\cite[Theorem 6.21]{V64}.

\begin{theorem}\label{thm:strictsepVR}
  Deciding whether two compact semialgebraic sets $A$ and $B$ in $\RN^d$ can be strictly separated by a hyperplane is \VR-complete.
\end{theorem}
\begin{proof}
 To see \VR-hardness we reduce from bounded polynomial feasibility, which is \ER-complete (see the proof of Theorem~\ref{thm:inconvER}). So we are given a polynomial $f: \RN^d \rightarrow \RN$ and are asked whether there is an $x \in B^d(0,1)$ such that $f(x) = 0$. Define $A = B^d(0,1)$ and $Z = \{x:\ f(x) = 0 \wedge x \in B^d(0,1)\} \cup B^d(p,1)$, where $p \in \QN^d$ is a point of distance strictly greater than $2$ from the origin. If $f$ is feasible, then $Z \cap A \neq \emptyset$ so $A$ and $Z$ are not separable (let alone strictly separable), otherwise, $A = B^d(0,1)$ and $Z = B^d(p,1)$ are disjoint convex sets which are strictly separated by a hyperplane.

 To see membership in \VR, suppose we are given two compact semialgebraic sets $A = \{x \in \RN^d:\ \alpha(x)\}$ and $B = \{x \in \RN^d:\ \beta(x)\}$. It follows from Kirchberger's theorem~\ref{thm:KBthm} that $A$ and $B$ are strictly separable by a hyperplane if and only if every $d+2$ points in $A$ and $d+2$ points in $B$ are sstrictly separated by a hyperplane. This allows us to  express that $A$ and $B$ are strictly separable by a hyperplane as
 \begin{align*}
 (\forall a_1, & \ldots, a_{d+2})(\forall b_1, \ldots b_{d+2})\ \\
              & \left(\bigwedge_{i \in [d+2]} \alpha(a_i) \wedge \beta(b_i)\right) \rightarrow \conv(a_1, \ldots, a_{d+2}) \cap \conv(b_1, \ldots, b_{d+2}) = \emptyset,
 \end{align*}
 where $\conv(P)$ denotes the convex hull of the points in $P$ and we use the fact that two compact convex sets are strictly separable if and only if they are disjoint~\cite[Theorem 2.10]{V64}. Since $\conv(a_1, \ldots, a_{d+2}) \cap \conv(b_1, \ldots, b_{d+2}) = \emptyset$ can be written
 as
  \begin{align*}
 (\forall \lambda_1, & \ldots, \lambda_{d+2})(\forall \mu_1, \ldots, \mu_{d+2} \in \RN_{\geq 0}) \\
         & \biggl(\sum_{i \in [d+2]} \lambda_i = 1 = \sum_{i \in [d+2]} \mu_i\biggl) \rightarrow \sum_{i \in [d+2]} \lambda_i a_i \neq \sum_{i \in [d+2]} \lambda_i b_i,
  \end{align*}
 strict separability by hyperplanes lies in \VR.
\end{proof}

Theorem~\ref{thm:strictsepVR} can easily be adapted to separability:

\begin{corollary}\label{cor:sepVR}
    Deciding whether two compact semialgebraic sets $A$ and $B$ in $\RN^d$ can be separated by a hyperplane is \VR-complete.
\end{corollary}
\begin{proof}
The hardness proof remains the same. For membership we need to make two changes: $(1)$ the version of
Kirchberger's theorem we used is stated for strict separation, but there is a version for non-strict separation, see~\cite[Theorem 6.22]{V64}, this version requires $2d+2$ points $P$ (rather than $d+2$ points) so that $A \cap P$ and $B \cap B$ are {\em separated}, rather than strictly separated, by a hyperplane; $(2)$ since we only require separation, rather than strict separation, we need to replace the condition that  $\conv(a_1, \ldots, a_{d+2})$ and $\conv(b_1, \ldots, b_{d+2})$ are disjoint with their interiors being disjoint~\cite[Theorem 2.9]{V64}. To compare the interiors only, we require the $\lambda_i$ and $\mu_i$ to lie in $\RN_{> 0}$ rather than $\RN_{\geq 0}$.
\end{proof}

There are results in the style of Kirchberger for other geometric figures such as hyperspheres, hypercylinders, etc.\ many due to Lay, see~\cite[Section 11]{E93}; we would expect that these lead to \VR-complete separation problems as well.

\section{\Krass\ and Star-Shapedness}\label{sec:Krass}

A set $S$ in $\RN^d$ is {\em star-shaped} if there is a point $q \in S$ which sees every point of $S$; here, two points {\em see} each other if the line-segment connecting them belongs to $S$. By definition, star-shapedness of a semialgebraic set lies in \EVR, since it can be expressed as
\[(\exists q \in \RN^d)(\forall p\in \RN^d)(\forall \lambda \in [0,1]\ p \in S \rightarrow \lambda p + (1-\lambda) q \in S.\]

The alternative characterization of star-shapedness we will work with was published by \Krass\ in 1946 in Russian~\cite{K46}.

\begin{theorem}[\Krass]\label{thm:Kthm}
A non-empty compact set $S \subseteq \RN^d$ is star-shaped if and only if every $d+1$ points in $S$ can see a common point in $S$.
\end{theorem}

Breen~\cite{B79} includes two examples justifying the compactness assumption in \Krass's theorem: an open disk with its center pointer removed shows that the set has to be closed, and---an example he takes from Hare and Kenelly~\cite{HK68}---the set $S = \cup_{n} T_n$, where $T_n =  \{(x,y):\ n-1 \leq y \leq n, x+y \geq n\}$, shows that the set has to  be bounded, see Figure~\ref{fig:teeth}.
$S$ is not star-shaped even though for every finite subset of $S$ there is a point (even a disk) inside $T_1$ that can see all points in the set. As Hare and Kenelly observe, $S$ does not contain a maximal star-shaped subset.

\begin{figure}[htbp]
    \centering
    \begin{pspicture}(-1,-1)(5,5)
        \pspolygon[fillstyle=solid,fillcolor=lightgray,linecolor=white](0,0)(4,0)(4,4)(0,4)
        \pspolygon[fillstyle=solid,fillcolor=white]%
        (0,0)(1,0)(0,1)(1,1)(0,2)(1,2)(0,3)(1,3)(0,4)%(1,4)(0,5)(1,5)(0,6)
        \psaxes[showorigin=false]{->}(0,0)(-0.5,-0.5)(4.5,4.5)[$x$,0][$y$, 90]
        \put(2.3,2.4){$S$}
    \end{pspicture}
    \caption{A closed set $S$ in $\RN^2$ in which every three points can see a common point, but $S$ is not star-shaped: every point in $S$ can see at most finitely many points on the $y$-axis. Example due to Hare and Kenelly~\cite{HK68}.}
    \label{fig:teeth}
\end{figure}

Using \Krass's theorem we can express the star-shapedness of a non-empty compact set as
\[(\forall p_1, \ldots, p_{d+1})(\exists q)\ \bigwedge_{i \in [d+1]} p_i \in S \rightarrow  \bigwedge_{i \in [d+1]} p_iq \subseteq S.\]
This only gives us a $\VEVR$-characterization of star-shapedness, since the conditions $p_iq \subseteq S$ hide a universal quantifier. Indeed, we do not know how to improve on this for general semialgebraic sets $S$, instead we will work with a more restricted type of semialgebraic set which we will introduce in Section~\ref{sec:starER}; before that, in Section~\ref{sec:ourKrass} we prove the version of \Krass's theorem we need for our result.

\subsection{\Krass's Theorem}\label{sec:ourKrass}

The material in this section closely follows the proof of \Krass's theorem as presented in Valentine's book on convex sets~\cite[Section 6]{V64}.

Let $S_x = \{y:\ xy \subseteq S\}$, the set of points in $S$ visible from $x$, known as the {\em $x$-star of $S$}. We call a point $x \in S$ {\em spherical} if there is an open ball $B$ disjoint from $S$ and such that $x$ lies on the boundary of $B$ (we take this terminology from Cel~\cite{C91}; Valentine's presentation works with a different type of points called regular). Note that a spherical point has to lie on $\partial S$. For a spherical point $x$ and such an open ball $B$, let $H$ be the closed half-space disjoint from $B$ which is tangential to $B$ at $x$; we call $H$ a {\em supporting half-space} at $x$. We have $S_x \subseteq H$; the reason is that any line-segment ending in $x$ must intersect $B$ or $H$. Let ${\cal H}$ be the set of all supporting half-spaces of $S$.

\begin{lemma}[\Krass]\label{lem:Klemma}
 Suppose $S \subseteq \RN^d$ is a closed set with $y \in S$, and $x \in \RN^d$ such that $xy \nsubseteq S$. Then there is a spherical point $z \in S$ such that $x \notin H$, where $H$ is a supporting half-space at $z$. In particular, $x \notin S_z$.
\end{lemma}

Valentine's argument~\cite[Lemma 6.2]{V64} establishes our version of the lemma, but we include a sketch of the proof. See Figure~\ref{fig:Krasslemma} for an illustration.

\begin{figure}[htbp]
    \centering
   \begin{pspicture*}(-0.5,0.5)(6,5)

        \put(1.318,0){
        \rput{-28}{
        \pspolygon[fillstyle=solid,fillcolor=lightgray,linecolor=white](0,-1)(-7,-1)(-7,10.5)(0,10.5)
\psline[linewidth=0.5pt](-0.025,-1)(-0.025,10.5)
}}

        \put(3.2,4.5){$H$}

        \put(-0.26,0.25){
        \pscurve(0,4)(0.5,4.5)(1,4)(3.5,3.5)(2,3)(4,2)(3,1)}
        \pscircle[linewidth=0.5pt](3.5,3.5){0.3}
        \put(3.8,3){$B$}
        \cnode*(3.22,3.65){.05}{z}\nput{150}{z}{$z$}
        \cnode*(0.5,2){.05}{y}\nput{135}{y}{$y$}
        \cnode*(5,4.25){.05}{x}\nput{135}{x}{$x$}
        \cnode*(3.5,3.5){.02}{c}
        \ncline[linewidth=0.5pt,linestyle=dashed]{x}{y}
    \end{pspicture*}
    \caption{Moving a ball $B$ along $xy$ until it touches the boundary of $S$ in the spherical point $z$. In the drawing we have $x \not\in S$, but $x \in S$ is allowed.}
    \label{fig:Krasslemma}
\end{figure}

\begin{proof}
Since $xy \nsubseteq S$, there must be a point on $xy$ which does not belong to $S$; since $S$ is closed, there even is a ball, with center on $xy$ such that the closure of the ball is disjoint from $S$. Move the center of that ball along $xy$ towards $y$ until the closure of the ball intersects $S$; this must happen, eventually, since $y \in S$. For the first such ball $B$ (a meaningful definition, since $S$ is compact), pick $z$ in the intersection of $S$ and $B$. Then $B$ witnesses that $z$ is a spherical point. Let $H$ be the supporting half-space at $z$. Since the center of $B$ lies on $xy$ and we were moving the ball towards $y$, the point $z$ must belong to the half of $B$ (as viewed along $xy$) that is closer to $y$; hence $x \notin H$, which completes the proof.
\end{proof}

\begin{theorem}[\Krass]\label{thm:ourKthm}
 Suppose $S \subseteq \RN^d$ is a non-empty compact set and ${\cal H}$ is the set of all supporting half-spaces of $S$. Then $S$ is star-shaped if and only if the intersection of every $d+1$ half-spaces (not necessarily distinct) in ${\cal H}$ is non-empty.
\end{theorem}

The following argument is based on the proof of~\cite[Theorem 6.17]{V64}; that proof assumes the connectedness of $S$, which is unnecessary.

\begin{proof}
If $S$ is star-shaped, let $c \in S$ be a point that sees all points of $S$. It follows that $c$ belongs to $S_x$ for every $x \in S$. If $H \in {\cal H}$, then there is a spherical point $z \in S$ for which $S_z \subseteq H$, so in particular, $c \in S_z \subseteq H$ and $c$ belongs to the intersection of all half-spaces in ${\cal H}$.

For the other direction, let $T$ be a minimal axis-aligned bounding-box for $S$ (such a box exists, since $S$ is compact). Then $T$ is the intersection of $2d$ half-spaces in ${\cal H}$, since each of the bounding half-planes must contain a point on the boundary of $S$. We are assuming that the intersection of every $d+1$ half-spaces in ${\cal H}$ is non-empty, so by the finite version of Helly's theorem, the intersection of every $2d+(d+1)$ half-spaces in ${\cal H}$ is non-empty. Since we can always include the $2d$ half-spaces whose intersection is $T$, we can conclude that every $d+1$ half-spaces in ${\cal H}$ contain a point in $T$. 

It follows that the intersection of
every $d+1$ sets in $\{H \cap T:\ H \in {\cal H}\}$ is non-empty. Since the sets $H \cap T$ are convex and compact, we can apply Helly's theorem to conclude that there is a point $x$ that belongs to all half-spaces in
${\cal H}$ (this argument works even if $\lvert{\cal H}\rvert \leq d$ since we do not require the half-spaces to be distinct). We claim that $x$ belongs to $S$ and sees all of $S$: Pick an arbitrary point $y \in S$ (this is possible since $S \neq \emptyset$). If $xy \nsubseteq S$, then, by
Lemma~\ref{lem:Klemma}, there is a spherical point $z \in S$ such that $x \notin H$, where $H$ is a supporting half-space at $z$. Since $z$ is spherical, $H$ belongs to ${\cal H}$, contradicting the definition of $x$. It follows that $xy \subseteq S$, so $x \in S$ and $x$ sees $y$. As $y$ was chosen arbitrarily in $S$ this implies that $x$ sees all of $S$, so $S$ is star-shaped.
\end{proof}

\subsection{Complexity of Star-Shapedness}\label{sec:starER}

We call a set $S \subseteq \RN^d$ a {\em region} if there is an integer polynomial $f \in \ZN[x]$ such that $S = \{x \in \RN^d:\ f(x)\geq 0\}$
and the boundary of the set $\partial S$ is $\{x\in \RN^d:\ f(x) = 0\}$. In general one only has $\partial S \subseteq \{x\in \RN^d:\ f(x) = 0\}$, consider, for example, $f = (1-x^2)x^2$ defining the interval $[-1,1]$ in $\RN^1$. We call the region {\em smooth} if the gradient $\nabla f$ does not vanish on $\partial S$, that is $\nabla f(x) \neq 0$ is a normal vector at $x$ for all $x \in \partial S$.

Requiring our semialgebraic sets to be smooth regions is quite restrictive, but this is the best notion for which we were able to prove a complexity upper bound for star-shapedness. We have to pay the price for the restrictiveness though; in the original version of this paper we claimed that the \VR-hardness for star-shapedness is similar to the one for convexity but it turns the details are somewhat technical and intricate.\footnote{The referees rightfully called us on this claim and we found out that restricting ourselves to smooth regions makes the lower bound significantly harder to prove than what we expected.} Since the focus on this paper is really on complexity upper bounds, we have split the proof of Theorem~\ref{thm:starVR} into two parts; we first present the
membership proof in Section~\ref{sec:starmem} followed by the lower bound in Section~\ref{sec:starlower}.

\subsection{$\forall\RN$-membership}\label{sec:starmem}

To show membership in \VR\ we will use Theorem~\ref{thm:ourKthm} to express star-shapedness of $S$ as follows: the intersection of any $d+1$ half-spaces in ${\cal H}$ contains a point in $T$, where ${\cal H}$ is the set of supporting half-spaces for $S$. A supporting half-space can be described by the center $c$ of the ball $B$ through which it is defined, and the spherical point $z$ in which $B$ intersects $S$. Since we assumed that $S$ is a smooth region, the gradient does not vanish on $\partial S$, the vector $cz$ must be a (scalar) multiple of $\nabla f(z)$.

We can therefore express star-shapedness as
\begin{align}\label{eq:star1}
(\forall & z_1, \ldots, z_{d+1} \in \RN^d) 
\biggl(
\bigwedge_{i=1}^{d+1} f(z_i)=0  \biggr) \rightarrow \bigcap_{i=1}^{d+1} H_i \neq \emptyset,
\end{align}
where $H_i = \{y:\ (y-z_i) \cdot \nabla f(z_i) \geq 0\}$, the half-space through $z_i$ that is tangent to $S$ at $z_i$. 
Rewriting the implication $p \rightarrow q$ as $\overline{p} \vee q$, we get
\begin{align}\label{eq:star2}
(\forall & z_1, \ldots, z_{d+1} \in \RN^d) 
\biggl(
\bigvee_{i=1}^{d+1} f(z_i)\neq 0 \biggr) \vee \biggl(\bigcap_{i=1}^{d+1} H_i \neq \emptyset\biggr).
\end{align}
Let us have a closer look at the second condition of the or-clause: $\bigcap_{i=1}^{d+1} H_i \neq \emptyset$. 
We can rewrite this condition as
\begin{equation}\label{eq:star3}
(\exists x \in \RN^d)\ \bigwedge_{i=1}^{d+1} (x-z_i) \cdot \nabla f(z_i)\geq 0
\end{equation}
(where $\cdot$ is the dot-product) which is equivalent to
\begin{equation}\label{eq:star4}
(\exists x \in \RN^d)\ \bigwedge_{i=1}^{d+1} \nabla f(z_i) \cdot x \geq z_i \cdot \nabla f(z_i).
\end{equation}
We can restate~\eqref{eq:star4} as $(\exists x \in \RN^d)\ Ax \geq b$, where $A = (\nabla f(z_1)|\cdots|\nabla f(z_{d+1}))^T$ and $b = (z_1\cdot \nabla f(z_1), \dots, z_{d+1}\cdot \nabla f(z_{d+1}))^T$. By Farkas' lemma, Lemma~\ref{lem:farkas}, this formula is equivalent to 
$(\forall y \in \RN^{d+1}_{\geq 0})\ Ay \neq 0
 \vee b^Ty \leq 0$ which, after expanding out $Ay$ and $b^Ty$ shows that $\bigcap_{i=1}^{d+1} H_i \neq \emptyset$ is equivalent to
\[(\forall y \in \RN^{d+1}_{\geq 0})\ \biggl(\bigvee_{i=1}^d \sum_{j=1}^{d+1}    (\nabla f(z_i))_j y_j\neq 0\biggr) \vee \biggl( \sum_{j=1}^{d+1}  (z_j\cdot \nabla f(z_i)) y_j \leq 0\biggr).\]
We can then replace the $\bigcap_{i=1}^{d+1} H_i \neq \emptyset$ condition in~\eqref{eq:star2} with this universally quantified formula.
At that point, we can move all universal quantifiers to the leading block of quantifiers, giving us the required universal characterization, thereby showing that star-shapedness lies in \VR.

\subsection{$\forall\RN$-hardness}\label{sec:starlower}

To show \VR-hardness of star-shapedness we need to reduce a \VR-complete problem to it; one would think this should be quite straightforward. We know that testing whether a polynomial $f$ contains a root in the unit ball $B^n(0,1)$ is \ER-complete (this is the bounded polynomial feasibility we used earlier). With this $f$ we can define 
\[S = \{(x_1, \ldots, x_{n+1}) \in \RN^{n+1}:\ f(x_1, \ldots, x_n) = 0 \wedge \sum_{i=1}^{n+1} x_i^2 = 1\}.\]
If $f$ does not contain a root in $B^n(0,1)$, then $S$ is empty; otherwise the points of $S$ lie on the $n+1$-dimensional sphere and do not form a star-shaped set. This proof sketch looks promising but has various defects: the empty set is not traditionally considered star-shaped, so the sketch does not give a reduction to star-shapedness. Even if the empty set were star-shaped, the set $S$ is not a smooth region, so the hardness proof would not match our membership proof.

Fixing these issues turns out to be somewhat intricate and technical. We work with a strengthened version of the bounded polynomial feasibility problem in which the function has no zeros outside of the unit ball; since this variant may be of interest in itself we state it as a separate result. 

\begin{theorem}\label{thm:bdfeasstrong}
 Deciding whether a given integer polynomial $f: \RN^n\rightarrow \RN$ has a root in $B^n(0,1)$ is \ER-complete even if the polynomial is non-negative and all its roots lie in $B^n(0,1)$.
\end{theorem}

\begin{proof}
    We start with the bounded feasibility problem for an integer polynomial $f:\RN^n\rightarrow\RN$. Recall that we can assume that $f$ is positive if it has no zero in $B(0,1)$. Using Lemma~\ref{lem:boundedzeros} we can construct a non-negative polynomial $\hat{f}$ such that $f$ has a root in $B^n(0,1)$ if and only if $\hat{f}$ has a root in $B^n(0,2)$. Moreover, $\hat{f}$ has no roots outside $B^n(0,2)$. Then 
    \[g(x_1, \ldots, x_n) := \hat{f}(2x_1, \ldots, 2x_n)\]
    is the polynomial we are looking for. 
\end{proof}

But we need more, we need the sets we construct to be smooth regions. We start with a polynomial that either is positive or has all its zeros in the unit ball, as in Theorem~\ref{thm:bdfeasstrong}. We then show how to
add a new zero (in either case) so that the resulting polynomial has either one or multiple zeros. We turn each zero into a region by perturbing the polynomial very slightly; for that we introduce two perturbation functions in Section~\ref{sec:perturb}; in the case of a single zero, the first perturbation function makes the resulting region convex and non-empty, therefore star-shaped. The second perturbation function ensures that all regions we create in both cases are smooth; for this the perturbations need to be sufficiently small so we do not move the polynomial past any of its critical values; to this end we prove an effective version of Sard's theorem, which we prove in Section~\ref{sec:Sard}.

%\subsubsection{Bounds and Perturbations}\label{sec:BaP}

%In Section~\ref{sec:perturb} we show how to construct polynomials that allow us to perturb another polynomial just slightly. The amount of necessary perturbation will be controlled by the result from Section~\ref{sec:Sard}

\subsubsection{An Effective Version of Sard's Theorem}\label{sec:Sard}

A {\em critical point} of $f$ is an $x$ for which the gradient of $f$ vanishes, that is, $\nabla f(x) = 0$.
A {\em critical value} of $f$ is a value $f(x)$ at a critical point $x$ of $f$. We need a lower bound on the positive critical values of an integer polynomial; that bound is achieved by the following lemma, which can be viewed as an effective version of Sard's theorem. 

\begin{lemma}\label{lem:lsepar}
Suppose that $f: \RN^n\rightarrow \RN$ is a non-negative integer polynomial. Then the positive critical values of $f$ are at least $2^{-2^{n^c}}$ for some constant $c$ depending on $f$ (effectively so); that is, $\min\{f(x):\ f(x) > 0 \wedge \nabla f(x) = 0\} > 2^{-2^{n^c}}$.
\end{lemma}

To prove the lower bound in the lemma we use an idea by B\"{u}rgisser and Cucker~\cite[Proof of Theorem~9.2]{BC09}: apply quantifier elimination followed by a univariate separation bound. We take the separation bound from B\"{u}rgisser and Cucker~\cite[Lemma 9.6]{BC09}:

\begin{lemma}\label{lem:unisep}
  If $h: \RN\rightarrow \RN$ is a non-constant univariate integer polynomial of degree $d$, then every two roots of $h$ have distance at least
  \[d^{-(d+2)/2}\tau^{-(d-1)},\]
  where $\tau$ is an upper bound on the coefficients of $h$.
\end{lemma}

\begin{proof}[Proof of Lemma~\ref{lem:lsepar}]
Let 
\[Z = \{ f(x) \in \RN:\ \nabla f(x) =0 \}\]
be the set of all critical values of $f$. Then $Z$ is semialgebraic, since it can be written as $\{z \in \RN:\ (\exists x\in \RN^n)\ f(x) - z = 0 \wedge \nabla f(x) = 0\}$ and $f$ is a polynomial, so we can compute $\nabla f$ explicitly. By the semialgebraic version of Sard's theorem~\cite[Theorem 5.6]{BPR06} the set $Z$ has dimension $0$, which implies that it is a finite set of points (using~\cite[Proposition~2.1.7]{BCR98} for example).  

Expressing $z \in Z$ as
\[(\exists x \in \RN^n)\ z -f(x) = 0 \wedge  \nabla f(x) =0\]
allows us to apply quantifier elimination to eliminate $x$ and obtain an equivalent formula of the form
\[\bigvee_{i=1}^I\bigwedge_{j=1}^{J_i} h_{ij}(z) \Delta_{ij} 0,\]
where $\Delta_{ij} \in \{<,\leq,=,>,\geq\}$, $\sum_{i=1}^I J_i \leq 2^{n^c}$ and the $h_{ij}$ are polynomials of degree at most $2^{n^c}$ with integer coefficients having at most $2^{n^c}$ bits, where $c$ is a fixed constant only depending on $f$ (see~\cite[Proof of Theorem~9.2]{BC09} or~\cite[Theorem 1.3.1]{BPR96}). Define
\[h(z) = \Pi_{i=1}^I \Pi_{j=1}^{J_i} h_{ij}(z).\]
Then $h$ has degree $d$ at most $2^{2n^c}$ with integer coefficients of bitlength at most $2^{2n^c}$, so by Lemma~\ref{lem:unisep}, the roots of $h$, and therefore the elements of $Z$ have distance at least $2^{-2^{n^{c'}}}$ from each other, for some constant $c'$; this implies the statement of the lemma.
\end{proof}

\subsubsection{The Perturbation Functions $g_{\alpha}$ and $\hat{g}$}\label{sec:perturb}

For $\alpha\in (-\infty,1]$ and $m \in \NN$ define the polynomial $g_{\alpha}: \RN^{m+1}\rightarrow\RN$ by
\[
g_\alpha (y) = g_\alpha (y_1, \ldots, y_{m+1}) = 400 \sum_{k\in [m+1]} (y_{k} - y_{k-1}^2)^2 + y_{m+1}^2 - \alpha y_{m}^2,
\]
where $y_0 =1/4$ is fixed. This polynomial has a minimum value just below $0$ and we will use it to slightly perturb other polynomials.

\begin{lemma}\label{lem:lupper} Let $g_{\alpha}$ be as defined above and $\alpha\in (-\infty,1]$.
\begin{description}
    \item[\bf Claim 1] We have
\begin{equation}\label{eqa1}
\min_{y\in \RN^{m+1}} g_\alpha(y) \in \Big[-(1/4)^{(5/6)2^m}, 
(1/4)^{2^{m+2}} - \alpha (1/4)^{2^{m+1}}\Big].
\end{equation}
\item[\bf Claim 2] If $g_\alpha(y)\leq 0$ then 
\begin{equation}\label{eq:bds}
(1/4)^{(7/6)2^k} < y_k < (1/4)^{(5/6)2^k}
\end{equation}
for all $k\in [m]$ and
\begin{equation}\label{eq:bds2}
(y_{m+1})^2 \leq y_m^2.
\end{equation}
\item[\bf Claim 3] If $\alpha\leq 0$ then $g_{\alpha}$ is positive.
\end{description}
\end{lemma}

\begin{proof}
We start by proving the lower bound in~\eqref{eqa1}.
If $\alpha \leq 0$ then $g_{\alpha}$ is non-negative, since it is a sum of squares. This immediately implies the lower bound in~\eqref{eqa1}. Hence, we can assume $\alpha\in (0,1]$. Note that
\[
\frac{\partial g_{\alpha}(y) }{\partial y_{m+1}} = -800y_m^2 + 802y_{m+1}.
\]
To prove the lower bound we can then assume that $y_{m+1}=(400/401)y_m^2$, since this minimizes $g_{\alpha}(y_1,\dots,y_{m+1})$ if the values of other $y_k$'s are fixed. 

Assume that we have $y_1,\dots,y_m,y_{m+1}$ such that $g_{\alpha}(y_1,\dots,y_{m+1}) \leq 0$. For this to happen we need to have
\begin{equation}\label{e0}
400 \sum_{k\in [m]} (y_{k} - y_{k-1}^2)^2 + \frac{400}{401} y_m^4 \leq y_m^2,
\end{equation}
using $\alpha \leq 1$. Dropping all but the last term on the left-hand side of~\eqref{e0} gives us $y_m^2 \leq 401/400$, and, in particular, $|y_m|\leq 3/2$.

If $y_m=0$ then $g_\alpha$ cannot be negative (since the remaining terms are a sum of squares) and it actually cannot be zero (it is impossible for all the squares to be zero). Thus $y_m\neq 0$.
Hence from~\eqref{e0} we can also conclude that
\begin{equation}\label{e1}
400 \sum_{k\in [m]} (y_{k} - y_{k-1}^2)^2 < y_{m}^2.
\end{equation}
In\cite[Lemma 2.7]{SS23,SS25} it is shown that~\eqref{e1} implies~\eqref{eq:bds}.\footnote{The statement of Lemma 2.7 requires $|y_m|\leq 1$ but the proof only uses the bound $|y_m|/2 + 1/10 < 1$. This is implied by $|y_m|\leq 3/2$.} Since we assume that $g_{\alpha}(y)$ is negative this also gives us~\eqref{eq:bds2}. This proves Claim~2 of the lemma. 
All terms but $y_m^2$ in $g_{\alpha}(y)$ are non-negative and hence 
$$g_{\alpha}(y) \geq -y_m^2 \geq -(1/4)^{(5/6)2^m}.$$
To show the upper bound in~\eqref{eqa1} take $y_k = (1/4)^{2^k}$ for $k\in [m+1]$. This proves Claim~1 of the lemma.

This leaves us with the proof of Claim~3. We already saw that $g_{\alpha}$ is non-negative for $\alpha \leq 0$ since in that case it is a sum of squares. If we had $g_{\alpha}(y) = 0$, then $y_{m+1} = 0$ which (inductively) forces  forces $y_k=0$ for all $k\in [m]$, however then the square $(y_1-1/16)^2$ is positive.
\end{proof}

We will also make use of the following variant of $g_{\alpha}$: 
\[
\hat{g} (y) =  g_1 (y) +  400 \sum_{k\in [m+1]} (y_{k} - y_{k-1}^2)^2.
\]

\begin{lemma}\label{lem:lelobo}
For the polynomial $\hat{g}$ as defined above the three claims of Lemma~\ref{lem:lupper} are true (using $\alpha=1$ in~\eqref{eqa1}). Moreover if $|y_i|\geq 1/10$ for some $i\in [m+1]$ then 
\begin{equation}\label{e111}
\hat{g} (y)  \geq 1/4.
\end{equation}
\end{lemma}

\begin{proof}
By definition $\hat{g} (y) \geq g_1 (y)$ and hence Claims~2, 3 and the lower bound in Claim~1 hold for $\hat{g}$. For the upper bound in Claim~1 note that the terms added to $g_1(y)$ are all zero for the point that witnessed the upper bound.

Suppose that $\hat{g} (y) < 1/4$. Using the lower bound on $g_1(y)$ from Lemma~\ref{lem:lupper} then gives us 
\begin{equation}\label{errr}
400 \sum_{k\in [m+1]} (y_{k} - y_{k-1}^2)^2 < 1/4+ (1/4)^{(5/6)2^m} \leq 1/2.
\end{equation}
If $|y_1|\geq \frac{1}{10}$ then $400(y_1 - 1/16)^2 \geq 9/16>1$, contradicting~\eqref{errr}. If for some $k\in [m]$ we have
$|y_k|\leq 1/10$ and $|y_{k+1}|\geq 1/10$ then $400(y_{k+1}-y_k^2)^2\geq 81/25>1$, also contradicting with~\eqref{errr}. Hence~\eqref{errr}
implies that  $|y_i| < 1/10$ for all $i\in [m+1]$.
\end{proof}

Lemma~\ref{lem:lelobo} implies that
\[ \{y\in \RN^{m+1}:\ \hat{g}(y) \leq 0\} \subseteq B^{m+1}_{\infty}(0,1/10), \]
(working with balls in the infinity-norm) and the same is true for $g_{\alpha}$:
\[ \{y\in \RN^{m+1}:\ g_{\alpha}(y) \leq 0\} \subseteq B^{m+1}_{\infty}(0,1/10), \]
by Claim 2 of Lemma~\ref{lem:lupper}, since $(1/4)^{2(5/6)}\leq 1/10$. 
Our next goal is to show that both $g_{\alpha}$ and $\hat{g}$ have unique minimal points in $B^{m+1}_{\infty}(0,1/10)$.

To this end we need the to look at second-order derivatives of $g_{\alpha}$ and $\hat{g}$. Recall that the {\em Hessian} $\nabla^2 f$ of a function $f$ is the (symmetric) matrix containing the second-order partial derivatives $\frac{\partial f}{\partial x_i \partial x_j}$. A matrix is {\em diagonally dominant} if each diagonal entry is larger than the absolute sum of the other entries in its row. A symmetric diagonally dominant matrix is positive definite (this follows from  Gershgorin's circle theorem, see~\cite[Section 6.1]{HJ13}). We will use the fact that a function whose Hessian is positive definite is strictly convex~\cite[Section~3.1.4]{BV04}.

\begin{lemma}\label{lem:ledipo}
Let $\alpha\in [-1,1]$. The Hessian of $g_\alpha(y)$ is  positive definite (even diagonally dominant) on $B^{m+1}_{\infty}(0,1/10)$.
% the set  $$S = \{ (y_1,\dots,y_{m+1}):\ (\forall i\in [m+1])\,|y_i|\leq 1/10 \}.$$ 
It follows that $g_\alpha(y)$ has a unique critical point in $B^{m+1}_{\infty}(0,1/10)$ and
\[ \{y\in \RN^{m+1}:\ g_{\alpha}(y) \leq 0\} \subseteq B^{m+1}_{\infty}(0,1/10). \]
The same claims are true for $\hat{g}$ in place of $g_{\alpha}$.
\end{lemma}

\begin{proof}
We already argued that $\{y\in \RN^{m+1}:\ g_{\alpha}(y) \leq 0\}$ and $\{y\in \RN^{m+1}:\ \hat{g}(y) \leq 0\}$ are contained in $B^{m+1}_{\infty}(0,1/10)$.

Let $H$ denote the Hessian of $g_{\alpha}/800$, so $H_{i,j} = \frac{\partial^2}{\partial y_i\partial y_j} (g_{\alpha}(y)/800)$. Then $H$ is a tri-diagonal matrix (that is only the entries $H_{i,i-1}, H_{i,i}, H_{i,i+1}$ are non-zero) with the following entries
\begin{equation}\label{eqh1}
H_{i,i} = \Bigg\{\begin{array}{ll}
1 - 2 y_{i+1} + 6 y_i^2 & i\in [m-1] \\
1 - 2 y_{i+1} + 6 y_i^2 - \alpha^2/400 & i=m \\
1+1/400 & i=m+1
\end{array}
\end{equation}
and
\begin{equation}\label{eqh2}
H_{i,i+1} = H_{i+1,i} = - 2 y_i\quad\mbox{for}\quad i\in [m].
\end{equation}
The bound $|y_i|\leq 1/10$ for $i\in [m+1]$ implies that $H$ is diagonally dominant:
\begin{equation}\label{eq:dido1}
H_{i,i} - \sum_{j\neq i} |H_{i,j}| \geq 1 - 6 (1/10) - 1/400 > 1/3,
\end{equation}
and hence positive definite.

The polynomial $g_{\alpha}$ achieves
its minimum in $B^{m+1}(0,1/10)$, since that set is compact, and the point is critical. A strictly convex, continuously differentiable function has at most one critical point on a convex set~\cite[Section~9.1.2]{BV04}.

For polynomial $\hat{g}$ the proof is essentially the same (looking at the Hessian of $\hat{g}/1600$ instead).
\end{proof}

We will use the function $g_\alpha(y_1,\dots,y_{m+1})$ to ``boost'' zeros of a non-negative polynomial $f$ to negative regions in such a way that for positive $f$
no negative regions will be created. The following result is a warm-up version (we will later see how to make the regions smooth by adding $\hat{g}$).

It is known that for every integer polynomial $f: \RN^n\rightarrow \RN$ we can effectively compute a constant $c$ such that if $f$ is positive in a compact domain, say $B^n(0,2)$ for our purposes, then $f(x) > \varepsilon_f := 2^{-2^{n^c}}$ for every $x \in B^n(0,2)$; this follows, for example, from~\cite[Theorem 1]{JP10} (which gives much finer bounds). 
%With this result we can define \[m(f) := \min\{m \in \NN:\ (1/4)^{(5/6) 2^m} < 2^{-2^{n^c}}\}.\]

\begin{lemma}\label{lem:boundedzeros}
Let $f:{\mathbb{R}}^n \rightarrow {\mathbb{R}}$ be a non-negative integer polynomial such that either $f$ has no zero in ${\mathbb{R}}^n$ or $f$ has at least one zero in $B^n(0,1)$.
Fix the smallest $m$ for which $(1/4)^{(5/6) 2^m} < \varepsilon_f$ and for $x \in \RN^n, y \in \RN^m$ define
 \[
\hat{f}(x,y) = f(x) + g_\alpha(y),
\]
where $\alpha=1-(x_1^2+\cdots+x_n^2)/2$.
%, $x = (x_1, \ldots, x_n) \in \RN^n$ and $y = (y_1, \ldots, y_m) \in \RN^m$.
%and $g_\alpha$ is the polynomial defined in Lemma~\ref{lem:lupper}.
%
Then the following are true:
\begin{itemize}
\item If $f$ has a root in $B^n(0,1)$ then $\hat{f}$ takes a negative value in $B^n(0,2)$.
\item If $f$ has no root in $B^n(0,1)$ then $\hat{f}$ has no root in $B^n(0,2)$.
\item $\hat{f}$ is positive outside $B^n(0,2)$.
\end{itemize}
\end{lemma}

\begin{proof}
Suppose $f$ has a root at $x = (x_1,\dots,x_n)\in B^n(0,1)$. Fix $y \in \RN^{m+1}$ such that $y_k = (1/4)^{2^k}$ for $k=1,\dots,m+1$. Then $\sum_{i=1}^n x_i^2 + \sum_{j=1}^{m+1} y_j^2 <2$ and
\[\hat{f}(x,y)\leq (1/4)^{2^{m+1}} - (1/2) (1/4)^{2^{m}} < 0\]
showing that $\hat{f}$ takes a negative value in $B^n(0,2)$.

Suppose that $f$ has no root in $B^n(0,1)$; by assumption on $f$ it is positive on $\RN^n$ and, in particular, has no root in $B^n(0,2)$. By definition of $m$
we have $(1/4)^{(5/6)2^m} < \min \{f(x):\ x\in B^n(0,2)\}$.
Then, by Claim 1 in Lemma~\ref{lem:lupper}, we have
$\hat{f}(x,y)>0$
for every $x \in B^n(0,2)$ and $y \in \RN^{m+1}$ so $\hat{f}$ has no root in $B^n(0,2)$.

If $x_1^2+\cdots+x_n^2\geq 2$ then the $\alpha$ used in the construction of $\hat{f}$ satisfies $\alpha\geq 0$. Then $g_\alpha(y)>0$ by Claim 3 in Lemma~\ref{lem:lupper}, and hence 
\begin{equation}\label{eqco}
    \hat{f}(x,y)>0. 
\end{equation} 
If $y_1^2+\cdots+y_{m+1}^2\geq 2$ then\eqref{eq:bds} and~\eqref{eq:bds2} cannot be satisfied and, by Claim 2 of Lemma~\ref{lem:lupper} we obtain 
$g_\alpha(y)>0$ and hence again~\eqref{eqco} is true. Combining the two claims we obtain that 
$x_1^2 + \cdots + x_n^2 + y_1^2 + \cdots + y_{m+1}^2\geq 4$ implies~\eqref{eqco}, which shows that if $(x,y) \not\in B^{n+(m+1)}(0,2)$, then $\hat{f}$ is positive, establishing the final claim. 
\end{proof}

\subsubsection{Proof of $\forall\RN$-hardness}

The proof of \VR-hardness of star-shapedness is completed in the next lemma by a reduction from the bounded polynomial feasibility problem. We note that the same reduction also establishes that testing convexity of smooth regions is \VR-hard; since convexity testing lies in \VR, indeed is \VR-complete~\cite{CR92,SS23}, we conclude that convexity testing of smooth regions is \VR-complete. 

\begin{lemma}
Let $f:{\mathbb{R}}^n \rightarrow {\mathbb{R}}$ be a non-negative polynomial such that either 1) $f$ has no zero in ${\mathbb{R}}^n$ or 2) $f$ has at least one zero in $B^n(0,1)$. We can efficiently construct 
a polynomial $\tilde{f}:{\mathbb{R}}^{n'} \rightarrow {\mathbb{R}}$ such that the set $S=\{x\in {\mathbb{R}}^{n'}:\ \tilde{f}(x)\leq 0\}$ is bounded, non-empty, and has $\nabla\tilde{f}(x)\neq 0$ for all $x\in\partial S$.
Moreover the set $S$ is convex in case 1 and $S$ is not star-shaped in case 2, since it will consist of at least two components. 
\end{lemma}

\begin{proof}
Let $\varepsilon_0 := \varepsilon_f$ with $\varepsilon_f$ as defined before Lemma~\ref{lem:boundedzeros} and let $\varepsilon = \min\{f(x):\ x \in B(0,1)\}$. In case 1 we have $\varepsilon > \varepsilon_0$, by definition. For $x = (x_0,x_1, \ldots, x_n) \in \RN^n$
\[
h(x) = (1-x_0)^2 f(x_1,\dots,x_n)  + x_0^2(1-x_0)^2  + x_0^2 \sum_{i=1}^n x_i^2.
\]
All terms in $h$ are non-negative and the last term is zero only if $x_0=0$ or if $\sum_{i=1}^n x_i^2=0$.
The second term is zero if and only if $x_0=0$ or $x_0=1$. Hence $h$ is zero if and only if either $x_0=1$ and $x_1=\cdots=x_n=0$
or $x_0=0$ and $f(x_1,\dots,x_n)=0$. 
Let us see how the assumptions on $f$ translate to $h$. 
In case 1 we have that $h$ has a unique zero at $x_0=1$, $x_1=\cdots=x_n=0$ and in case 2
we have that $h$ has at least two zeros in $B^{n+1}(0,1)$, one with $x_0=1$ and at least one with $x_0=0$. 

We will use the following lower bound on $h$. For $x\in B^{n+1}(0,1)$ we have
\begin{equation}\label{eq:lo}
h(x) \geq x_0^2 \sum_{i=1}^n x_i^2 + (1-x_0)^2\varepsilon.
\end{equation}

%Next we take the polynomial from Lemma~\ref{lem:boundedzeros} that ``boosts'' zeros of $h$ into negative regions and  is positive outside $B^n(0,2)$ (so that the regions are compact). 

We use Lemma~\ref{lem:boundedzeros} to ``boost'' the zeros of $h$ into negative regions and such that there are no zeros outside $B^{(n+1)+(m+1)}(0,2)$ so that all regions are compact. To this end let
\begin{equation}\label{eqM}
M = \max\left(\max_{i,x\in B^{n+1}(0,2)} \left|\frac{\partial}{\partial x_i} f(x)\right|, \max_{i,j,x\in B^{n+1}(0,2)} \left|\frac{\partial}{\partial x_j\partial x_i} f(x)\right|\right),
\end{equation}
and fix $m \in \NN$ such that 
\begin{equation}\label{eqdelta}
\delta:=(1/4)^{(5/6) 2^m} < (\varepsilon_0)^8 \left(48(n+1)(M+1)\right)^{-4}.
\end{equation}
Note that $2\delta/\varepsilon_0 < \sqrt{\delta}$. 

With $y = (y_1, \ldots, y_{m+1}) \in \RN^{m+1}$ define
\[
\hat{f}(x,y) = h(x) + g_\alpha(y),
\]
where $\alpha= 1 - (x_0^2 + x_1^2 + \cdots + x_n^2)/2$. 

We ``smooth'' the polynomial so that the boundaries of the regions have non-zero gradient. Let $m'\geq m$ be such that 
$(1/4)^{(5/6)2^{m'}}$ is smaller than the lower bound from  Lemma~\ref{lem:lsepar} on the non-zero critical values of 
$\hat{f}$. With $z = (z_1,\ldots,z_{m'+1}) \in \RN^{m'+1}$ define
\begin{align*}
\tilde{f}(x,y,z) &= 
    \hat{f}(x,y) + \hat{g}(z)\\
    &=
    h(x)+g_{\alpha}(y)+\hat{g}(z).
\end{align*}
Let $S$ consist of the points $(x,y,z)\in\RN^{n+m+m'+3}$
%$ \in {\mathbb{R}}^{(n+1)+(m+1)+(m'+1)}$ 
for which $\tilde{f}(x,y,z)\leq 0$. 

Assume that we are in case 1, that is,  $h$ has a unique zero at $x_0=1$, $x_1=\cdots=x_n=0$.  We will show that $S$
consists of one convex component.  Note that for any $y$'s, $z$'s, and $\alpha\leq 1$ we have
$$
g_\alpha(y) + \hat{g}(z) > -2\delta. 
$$
Hence for $\tilde{f}$ to be negative we need to have $h(x)< 2\delta$, so by using~\eqref{eq:lo} we must have
$x_0^2 \sum_{i=1}^n x_i^2 + (1-x_0)^2\varepsilon<2\delta$. It follows that
\begin{equation}\label{eclose}
(1-x_0)^2\leq 2\delta/\varepsilon \quad\mbox{and}\quad x_i^2\leq 4\delta\quad \mbox{for all $i\in [n]$}.
\end{equation}
The first condition is immediate, and since $0< \delta < \varepsilon_0/48 < \varepsilon/48$ by~\eqref{eqdelta}, the first condition implies that $(1-x_0)^2 \leq 1/24$; in particular, $x_0^2 \geq 1/2$, which then allows us to conclude that $x_i^2 \leq 4\delta$, since $x_0^2x_i^2 < 2\delta$.

We now show that the Hessian of $\tilde{f}$ is positive definite at points where the function has negative value by showing that the Hessian is diagonally dominant. 

The Hessian has a block matrix structure with one block for $z$ and one block for the remaining variables (since in $\tilde{f}$ there are no monomials that contain a $z$-variable together with a non-$z$-variable). The Hessian for the block corresponding to $z$ is positive definite by Lemma~\ref{lem:ledipo}: the block is given by~\eqref{eqh1} and~\eqref{eqh2}, and the diagonal dominance of this block is established by~\eqref{eq:dido1} (with $z$ taking the role of $y$).

Let us look at the remaining block. The terms that combine $x$'s and $y$'s are going to have minimal influence on the Hessian; to make this claim precise it is helpful to first look at the submatrices indexed by $y$'s
and $x$'s separately. 
The Hessian of the block for $y$ is positive definite by Lemma~\ref{lem:ledipo}: the block is given by~\eqref{eqh1} and~\eqref{eqh2}, and the diagonal dominance of this block is established by~\eqref{eq:dido1}. In the Hessian of $\tilde{f}$ in the submatrix for $x$ we have
\begin{subequations}
\begin{align}[left = {\dfrac{\partial^2 \tilde{f}}{\partial x_i \partial x_j}  =  \empheqlbrace}]
&2 + 12 x_0 (x_0-1)  + 2 f(\check{x}) + 2\sum_{i=1}^n x_i^2 & i = j = 0 \label{ehes1} \\
&2 x_0^2 + (1-x_0)^2  \frac{\partial^2}{\partial x_i^2}  f(\check{x}) & i=j \in [n] \label{ehes2} \\
&4 x_0 x_i - 2(1-x_0) \frac{\partial}{\partial x_i} f(\check{x}) & i \in [n] \label{ehes3} \\
&(1-x_0)^2 \frac{\partial^2}{\partial x_i \partial x_j}  f(\check{x}) & i\neq j \in [n]
\end{align}
\end{subequations}
where we write $\check{x}$ for $(x_1, \ldots, x_n)$ that is $x$ without $x_0$.

\begin{comment}
\begin{equation}\label{ehes1}
\frac{\partial^2}{\partial x_0^2} \tilde{f} = 2 + 12 x_0 (x_0-1)  + 2 f(x_1,\dots,x_n) + 2\sum_{i=1}^n x_i^2,
\end{equation}
\begin{equation}\label{ehes2}
\frac{\partial^2}{\partial x_i^2} \tilde{f} = 2 x_0^2 + (1-x_0)^2  \frac{\partial^2}{\partial x_i^2}  f(x_1,\dots,x_n)\quad\quad\mbox{for $i \in [n]$}
\end{equation}
\begin{equation}\label{ehes3}
\frac{\partial^2}{\partial x_0 x_i} \tilde{f} = 4 x_0 x_i - 2(1-x_0) \frac{\partial}{\partial x_i} f(x_1,\dots,x_n)
\end{equation}
\begin{equation}\label{ehes4}
\frac{\partial^2}{\partial x_i x_j} \tilde{f} = (1-x_0)^2 \frac{\partial^2}{\partial x_i \partial x_j}  f(x_1,\dots,x_n), i\neq j.
\end{equation}
\end{comment}

Condition~\eqref{eclose} implies that~\eqref{ehes1} and~\eqref{ehes2} are close to $2$ and that~\eqref{ehes3} and~\eqref{ehes4} 
are close to $0$ and the submatrix for $x_0,x_1,\dots,x_n$ is diagonally dominant. More precisely for any $i$ we have
\begin{equation}\label{eq:dido2}
\frac{\partial^2}{\partial x_i^2}\tilde{f} - \sum_{j\neq i} \Big|\frac{\partial^2}{\partial x_i \partial x_j}\tilde{f}\Big|\geq
1 - (24+6(n+1))(1+M)\delta^{1/4} > 1/3.
\end{equation}
(To establish~\eqref{eq:dido2} we are using very loose bounds $|x_0|\leq 2$, $x_0^2 \geq 1/2$, $|x_0-1|\leq\delta^{1/4}$ and~\eqref{eqM}, 
\eqref{eqdelta}.)

The only entries in the Hessian that are outside of the submatrix for $y_1,\dots,y_{m+1}$ and the submatrix for $x_0,x_1,\dots,x_n$
are 
\begin{equation}\label{ehes4}
\frac{\partial^2}{\partial y_m \partial x_i} \tilde{f} = 2 y_m x_i.
\end{equation}
Note that $|y_m x_i|<1/3$ (using~\eqref{eq:bds} and~\eqref{eq:bds2})---it is smaller than the right-hand-side of equations
\ref{eq:dido1} and \ref{eq:dido2}. Hence the Hessian of~$\tilde{f}$ is diagonally dominant on the set $S$ of points for which $\tilde{f}\leq 0$. 
This implies that the set $S$ is convex (for any two points $a,b\in S$ and any point $c$ on the segment $a,b$ by convexity of $\tilde{f}$ the value 
of the function $\tilde{f}$ at $c$ must be $\leq 0$ and hence $c\in S$).

In case 2 we have that $h$ has at least two zeros, one for $x_0=1,x_1=\cdots=x_n=0$ and one for $x_0=0$ and $x\in B^n(0,1)$. The function $\tilde{f}$
will be negative for these points (for $y$'s and $z$'s that minimize $g_{\alpha}$ and $\hat{g}$). This creates at least $2$ disconnected components of points where $\tilde{f}$
is $\leq 0$ (since for $x_0=1/2$ the value of $h$ is at least $1/16$ and hence $\tilde{f}$ is positive). 

It only remains to establish regularity of $\tilde{f}$. Assume, for the sake of contradiction that $\tilde{f}$ is not regular,  that is, we have $\tilde{f}=0$ 
and $\nabla\tilde{f}=0$ at the same point $(x,y,z) \in \RN^{n+m+m'+3}$. 
If $|z_i|\geq 1/10$ for some 
$i\in [m'+1]$ then, by Lemma~\ref{lem:lelobo}, $\tilde{f}$ is positive
(since $g_\alpha$ has value at least $-1/100$ and all other terms in $\tilde{f}$ are non-negative).  Hence we must have $z \in B^{m'+1}_{\infty}(0,1/10)$.
Note that $\nabla\tilde{f}$ contains $\nabla \hat{g}(z)$ as a sub-vector and the condition $\nabla\tilde{f}=0$  implies that $z$
is the unique critical point of $\hat{g}$ in the $B^{m'+1}_{\infty}(0,1/10)$. That means $\hat{g}$ has a small negative value which, since $\tilde{f}=0$,
means $h+g_\alpha$ has a small positive value. Note that $\nabla\tilde{f}$ contains $\nabla (h+g_\alpha)(x,y)$ as a sub-vector (since
the $\hat{g}$ part does not depend on the $x$- and $y$-variables). The condition $\nabla\tilde{f}=0$ implies that $(x,y)$ is a critical point
of $h+g_\alpha$. This together with the small positive value of $h+g_\alpha$ contradicts Lemma~\ref{lem:lsepar}.
\end{proof}

\section{Open Questions}

We have shown that results from convex geometry, sometimes together with duality, can be used to settle the computational complexity of various properties of semialgebraic sets---at lower levels than might have been expected. Are there more examples like this? E.g.\ Helly's theorem can be used to show that testing whether the intersection of a parameterized family of convex sets is non-empty lies in \VER, is the problem complete for this level? And Kirchberger's theorem extends to separations by surfaces other than hyperplanes, such as hyperspheres, see~\cite{E93}; do these results leads to new complexity results? Also, as a referee points out, there are colorful versions of many of the results from convex geometry, do any of these yield interesting complexity characterizations? %Also: colorful versions, Tverberg's theorem

Our main result, Theorem~\ref{thm:starVR} on star-shapedness, is restricted to what we called smooth regions; this leaves unanswered the question of how hard it is to test star-shapedness of semialgebraic sets in general; while there does not seem to be an intrinsic reason to believe that star-shapedness becomes truly harder in this case, proving so would require much more sophisticated tools than the ones we used here. (Also: does testing star-shapedness remain in \VR\ if we drop the compactness assumption?)

Resolving this issue is relevant to determine the complexity of testing cone-shapedness. A set $S$ is {\em cone-shaped} if there is a point $a$ such that $(0,\infty)\cdot (S-a) \subseteq S-a$, see~\cite[Section 3]{HHMM20}. By definition, cone-shapedness lies in \EVR. Smooth regions are not cone-shaped though, so one needs to extend the methods for star-shapedness to work for non-smooth regions as well. Otherwise, cone-shapedness seems a promising candidate for our approach; since there is a \Krass-style characterization due to Cel~\cite{C91}. He showed that a set $S$ in $\RN^d$ is cone-shaped if and only if any $d+1$ spherical points (in our definition) are $S$-visible from a common point, where being {\em $S$-visible} means that a ray starting at the common point through the spherical point lies in $S$.

A similar issues arises for other properties; e.g.\ closedness and compactness are known to be \VR-complete for {\em basic} semialgebraic sets~\cite[Lemma 6.20, Corollary 9]{BC09}, but the complexity for semialgebraic sets remains open. Another example we mentioned earlier: the art gallery problem;  membership of the art gallery problem in \ER\ is based on the underlying set being a polygon; without this restriction, the current best upper bound on the art gallery problem is \VER, and it is not clear whether the universal quantifier can be eliminated in this case. 

One could attack the art gallery problem by looking at the case in which only a constant number $k$ of guards are allowed; this notion is known as $k$-star-shapedness and has been investigated in the literature, e.g.\ Koch and Marr~\cite{KM66} and Larman~\cite{L67} proved \Krass-style characterization of $2$-star-shaped sets; there seem to be no \Krass-style characterizations for larger $k$ (though there is a related result for $k = 3$ in the plane~\cite{B89}).

The polygonal vs semialgebraic issue in the art gallery problem also suggests the question of how hard it is to test whether a semialgebraic set is polygonal or polytopal in general. Dumortier, Gyssens, and Vandeurze~\cite{DGVVG99,DGVVG99b} call a semialgebraic set {\em semilinear} if all the defining equalities and inequalities are linear. They show that recognizing semilinearity lies in \VEVR; is this best possible? Is the problem hard for \ER\ or \VR\ or higher levels? And how about testing whether a semialgebraic set is a simplex?

Finally, star-shapedness itself deserves a closer look; it exists in many variants, see~\cite{HHMM20} for a recent survey, some with \Krass-style characterizations and these may make good candidates for further examples of unexpected complexity classifications. We'd like to explicitly mention a result by Valentine~\cite{V70} on {\em shoreline visibility}: if every five points in a plane set $S$, the {\em swimmers}, see a common boundary point of a compact convex set $C$ in $\RN^2$, the {\em coast}, then all points in $S$ can see a common boundary point of $C$ (and this remains true with $d+1$ points in $\RN^d$ if the boundary is smooth and $C$ is strictly convex). Can this result be used to determine the computational complexity of shoreline visibility for $S$ and $C$? By definition, shoreline visibility for semialgebraic sets $S$ and $C$ belongs to \EVR, and Valentine's result immediately implies a \VER-characterization (assuming $C$ is a smooth region and strictly convex), so that the problem belongs to $\EVR\cap\VER$; can this upper bound be improved?

\section*{Acknowledgments}

We would like to thank the anonymous reviewers whose feedback has allowed us to fix several errors and improve the overall presentation.

\bibliographystyle{plainurl}
\bibliography{stars}

\end{document}